\documentclass[11pt]{article}

\usepackage[margin=1in]{geometry}
\usepackage[T1]{fontenc}
\usepackage{amsmath,amssymb,amsthm}
\usepackage{newtxtext}
\usepackage{newtxmath}
\usepackage{mathtools}
\usepackage{bm}
\usepackage{microtype}
\usepackage{booktabs}
\usepackage[authoryear,round]{natbib}
\usepackage{authblk}
\usepackage[colorlinks=true,linkcolor=blue,citecolor=blue,urlcolor=blue]{hyperref}

\hypersetup{
  pdftitle={On the Convex Transform Order of Order Statistics},
  pdfauthor={Yijia Li, Maochao Xu, Peng Zhao}
}

\allowdisplaybreaks

\theoremstyle{plain}
\newtheorem{theorem}{Theorem}
\newtheorem{lemma}{Lemma}
\newtheorem{proposition}{Proposition}
\newtheorem{corollary}{Corollary}

\newcommand{\R}{\mathbb{R}}
\newcommand{\Pp}{\mathbb{P}}
\newcommand{\Ee}{\mathbb{E}}
\newcommand{\bc}{\boldsymbol{\alpha}}
\newcommand{\qbar}{\overline{\boldsymbol q}}

\numberwithin{equation}{section}

\title{On the Convex Transform Order of Order Statistics}

\author[1]{Yijia Li}
\author[2]{Maochao Xu}
\author[3]{Peng Zhao}

\affil[1]{School of Statistics and Data Science, Shanghai University of International Business and Economics, Shanghai, China}
\affil[2]{Department of Mathematics, Illinois State University, Normal, IL 61790-4520, USA\footnote{\textit{Corresponding author:} \href{mailto:mxu2@ilstu.edu}{mxu2@ilstu.edu.}}\\}
\affil[3]{School of Mathematics and Statistics, Jiangsu Normal University, Xuzhou, Jiangsu, China}

\date{\today}

\begin{document}
\maketitle

\begin{abstract}
Let $X_1,\ldots,X_n$ be independent exponential random variables with arbitrary positive, not necessarily equal, rates, and let $Y_1,\ldots,Y_n$ be independent and identically distributed exponential random variables. We prove that, for every order statistic, the homogeneous order statistic is smaller than its heterogeneous counterpart in the convex transform order. Our results show that, relative to the homogeneous benchmark, heterogeneity stretches upper quantiles more strongly than lower quantiles, while the homogeneous order statistic ages faster in the convex-transform sense. We further extend the comparison to a proportional-hazards family that includes Weibull, Lomax, and Burr XII distributions. We give explicit conditions under which the comparison is preserved and show that the shape restriction is sharp within the Weibull family. These results characterize how component heterogeneity changes the distributional shape of order statistics beyond effects on location or scale.
\end{abstract}



\section{Introduction}\label{sec:intro}
Distributional shape is a fundamental feature in probability and statistics. For nonnegative random variables, this shape is often reflected in the right tail. Two distributions may have similar means or comparable central behavior, yet differ substantially in their upper quantiles.   A natural way to compare distributional shape is the convex transform order of \cite{vanZwet1970}. If \(Y\) and \(X\) have distribution functions \(G\) and \(F\), respectively, then \(Y\le_c X\) means that \(F^{-1}\circ G\) is convex, or equivalently that \(G^{-1}\circ F\) is concave, where \(\le_c\) denotes the convex transform order. In this sense, the calibration from the quantile scale of \(G\) to that of \(F\) becomes steeper toward the upper tail. So the upper quantiles of \(F\) are stretched more strongly than the lower quantiles relative to \(G\). For lifetime distributions, this order is also closely related to relative ageing and IFR-type comparisons; see \cite{ShakedShanthikumar2007}. 

This paper studies this shape comparison for order statistics. Order statistics arise in a broad range of probabilistic and statistical problems, including tail-index estimation and extreme-value inference \cite{DreesFerreiraDeHaan2004}, distributional modeling of ordered samples \cite{ButuceaDelmasDutfoyFischer2018}, probabilistic representations based on exponential samples \cite{BrussGrubel2003}, and recent studies of extreme order statistics and their limit behavior \cite{PitmanTang2022}. In threshold systems, a system event occurs after a prescribed number of component events. If $X_i$ denotes the event time of component $i$, then order statistic
$X_{k:n}$ is the time at which the $k$th event occurs. Thus $X_{k:n}$ may
represent the time until $k$ recovery resources have been restored, or the
time until a quorum of $k$ responses has been received. In classical
reliability terminology, if $X_i$ denotes the lifetime of component $i$, the
lifetime of a $k$-out-of-$n$:G system is $X_{n-k+1:n}$; see, for example,
\cite{BarlowProschan1981}. The same threshold structure is relevant to
business-interruption risk when service can resume after a prescribed number
of operationally distinct recovery resources become available.
Component times in these systems are rarely identical. Machines may fail at
different rates, recovery resources may have different capabilities, and
participating units may respond at different speeds. Nevertheless, a homogeneous independent and identically distributed (i.i.d.) model is often used as an analytical benchmark. Such a model can be calibrated to match an overall scale, such as the mean or a characteristic time, but this does not determine whether it also captures the shape of the threshold-time distribution. The main question is therefore whether heterogeneity merely changes the level of an order statistic or also changes how its upper quantiles are stretched relative to the lower ones. In terms of the convex transform order, this asks whether the homogeneous order statistic provides a scale-free shape reference for its heterogeneous counterpart. In reliability theory, this is a relative-aging question \cite{LaiXie2006,kochar2022stochastic}; related failure-rate properties of parallel systems have also been studied by \cite{ArabHadjikyriakouOliveira2020Failure}. In recovery and insurance settings, the same shape comparison is related to upper-tail behavior and stop-loss risk \cite{kaas2008modern}, and transform orders have been used to compare largest claim amounts in heterogeneous insurance portfolios \cite{Zhang2021Transform}.

The effect of component heterogeneity on threshold systems has long been
studied in reliability theory. In an early contribution,
\cite{BolandProschan1983} considered $k$-out-of-$n$ systems with unequal
component reliabilities. They identified regions in which the system
reliability is Schur-convex or Schur-concave in the vector of component
reliabilities. Their results show that heterogeneity can improve or worsen
pointwise system reliability, depending on the threshold and on the component
reliability levels. 
The question considered here is different. We compare the shape of the full
threshold-time distribution, not only the reliability at a fixed time. For the maximum of independent heterogeneous exponential random variables,
\cite{KocharXu2009} proved a convex-transform comparison with the
corresponding i.i.d. exponential benchmark. The special role of the homogeneous benchmark was later clarified by \cite{ArabHadjikyriakouOliveira2020}, who gave general criteria for noncomparability under the convex transform order. They also showed that the Kochar--Xu comparison does not extend to two arbitrary heterogeneous parallel systems. For general exponential order statistics, \cite{Yu2021} established the corresponding heterogeneous-versus-homogeneous comparison under the star order. More recently, \cite{ArabLandoOliveira2025} developed inequalities and bounds for expected order statistics from transform-ordered families, further illustrating the role of transform orders in the study of order statistics. Since the star order is weaker than the convex transform order, the result of \cite{Yu2021} still left open whether the convex-transform comparison extends to every order statistic.

This paper resolves that question. Let $X_1,\ldots,X_n$ be independent
exponential random variables with arbitrary positive rates, and let
$Y_1,\ldots,Y_n$ be i.i.d. exponential random variables. We prove that, for
every $k=1,\ldots,n$,
\begin{equation}\label{eq:intro-main}
        Y_{k:n}\le_c X_{k:n},
\end{equation}
where $\le_c$ denotes the convex transform order. We further extend the comparison beyond exponential components. Specifically,
we consider the proportional-hazards family \cite{kochar2007stochastic}
\begin{equation}\label{eq:family-intro}
\overline F_{\lambda,\sigma,\alpha,\xi}(x)=
\begin{cases}
\displaystyle\left[1+\xi(x/\sigma)^\alpha\right]^{-\lambda/\xi},
    &\xi>0,\\[5pt]
\displaystyle\exp\{-\lambda(x/\sigma)^\alpha\},
    &\xi=0,
\end{cases}
\qquad x\ge0 .
\end{equation}
When $\xi=0$, this is the Weibull family, with the exponential distribution
as the case $\alpha=1$. When $\xi>0$, it includes Lomax and Burr XII
distributions with regularly varying tails \cite{KleiberKotz2003}. For \(0<\alpha\le1\), we give explicit rate conditions under which the convex-transform ordering is preserved. We also show that, when \(\alpha>1\), the ordering does not hold in general.

The rest of the paper is organized as follows. Section~\ref{sec:exp} proves the exponential comparison by reducing it to a boundary matrix inequality. Section~\ref{sec:matrix} establishes this key inequality. Section~\ref{sec:extensions} develops the proportional-hazards extension. Section~\ref{sec:discussion} concludes the paper and discusses possible extensions.

\section{Main results}\label{sec:exp}
 Let $X_1,\ldots,X_n$ be independent exponential random variables with rates $\lambda_i>0$, respectively, and let $Y_1,\ldots,Y_n$ be i.i.d. exponential with rate $\lambda_0>0$. For $k\in\{1,\ldots,n\}$, denote
\(
        m=n-k+1.
\)
For ${\bm q}=(q_1,\ldots,q_n)\in[0,1]^n$, let $B_1,\ldots,B_n$ be independent Bernoulli variables with $\Pp(B_i=1)=q_i$, and put $N=\sum_{i=1}^n B_i$. Define
\begin{equation}\label{eq:Dm}
D_m(z)=\Pp\left(\operatorname{Bin}(n,z)\ge m\right)
      =\sum_{j=m}^n\binom{n}{j}z^j(1-z)^{n-j},
\end{equation}
where \(\operatorname{Bin}(n,z)\) denotes a binomial random variable with \(n\) trials and success probability \(z\). To connect the heterogeneous Bernoulli probabilities ${\bm q}$ with a homogeneous benchmark, suppose that
$
        0<\Pp(N\ge m)<1.
$
Since \(D_m(z)\) is strictly increasing in \(z\in(0,1)\), there is a unique
\(z\in(0,1)\) such that
\begin{equation}\label{eq:calibration}
\Pp(N\ge m)=D_m(z).
\end{equation}
That is, \(z\) is the common Bernoulli success probability that gives the same tail probability as the heterogeneous vector ${\bm q}$. To facilitate the discussion, define
\begin{equation}\label{eq:bm}
b_m=\binom{n}{m}z^m(1-z)^{n-m}.
\end{equation}
For \(i\ne j\), let
$$
        N_{-i}=\sum_{\ell\ne i}B_\ell,
        \qquad
        N_{-ij}=\sum_{\ell\ne i,j}B_\ell.
$$
Define
\begin{equation}\label{eq:alpha}
\alpha_i=q_i\Pp(N_{-i}=m-1),
\qquad
\bm{\alpha}=(\alpha_1,\ldots,\alpha_n)^{\mathsf T}.
\end{equation}
Finally, let \(A=(a_{ij})\) be the symmetric matrix with entries
\begin{equation}\label{eq:Adef}
a_{ii}=\alpha_i,
\qquad
a_{ij}
=
q_iq_j
\left[
\Pp(N_{-ij}=m-2)-\Pp(N_{-ij}=m-1)
\right],
\quad i\ne j.
\end{equation}
The following proposition plays a key role in the proof of our main result. Its proof is deferred to Section~\ref{sec:matrix}.

\begin{proposition}\label{prop:boundary}
For \({\bm q}\in(0,1)^n\) and \(z\) determined by Eq.~\eqref{eq:calibration},
\begin{equation}\label{eq:boundary}
A-\frac{m-nz}{m b_m(1-z)}
\bc\bc^{\mathsf T}\succeq0,
\end{equation}
where \(\succeq0\) denotes positive semidefiniteness.
\end{proposition}

We now use Proposition~\ref{prop:boundary} to establish the main exponential comparison.

\begin{theorem}\label{thm:exp}
Let \(X_1,\ldots,X_n\) be independent exponential random variables with arbitrary positive rates \(\lambda_1,\ldots,\lambda_n\), and let \(Y_1,\ldots,Y_n\) be i.i.d. exponential random variables with common positive rate \(\lambda_0\). Then
\begin{equation}\label{eq:exporder}
Y_{k:n}\le_c X_{k:n},
\qquad k=1,\ldots,n.
\end{equation}
\end{theorem}

\begin{proof}
Let \(F_H\) and \(F_0\) be the distribution functions of \(X_{k:n}\) and \(Y_{k:n}\), respectively. For \(x>0\), define

$$
        q_i(x)=e^{-\lambda_i x},
$$
and let \(N_x\) be the sum of independent Bernoulli random variables with success probabilities ${\bm q}(x)=(q_1(x),\ldots,q_n(x))$. Set
$$
        S(x)=\Pp(N_x\ge m).
$$
Since \(X_{k:n}>x\) exactly when at least \(m=n-k+1\) components exceed \(x\),
$
        \overline F_H(x)=S(x).
$
For the homogeneous sample,
$$
        \overline F_0(x)=D_m(e^{-\lambda_0x}).
$$
For each \(x>0\), let \(z(x)\in(0,1)\) be determined by
\begin{equation}\label{eq:zx}
           D_m(z(x))=S(x). 
\end{equation}
Then
\begin{equation}\label{eq:quantilemap}
F_0^{-1}(F_H(x))
=
-\lambda_0^{-1}\log z(x).
\end{equation}
Hence it is enough to show that
$
T(x)=-\log z(x)$ is concave.

From Eq.~\eqref{eq:zx}, it holds that
\[
        D_m'(z(x))z'(x)=S'(x),
\]
and hence
\[
        T'(x)=-\frac{z'(x)}{z(x)}
        =-\frac{S'(x)}{z(x)D_m'(z(x))}.
\]
A second differentiation gives
\[
\begin{aligned}
T''(x)
&=
-\frac{S''(x)H(z)-S'(x)H'(z(x))z'(x)}{H(z(x))^2}  \\
&=
-\frac{S''(x)}{zD_m'(z(x))}
+
\frac{(S'(x))^2(D_m'(z(x))+zD_m''(z(x)))}
     {D_m'(z(x))(zD_m'(z(x)))^2},
\end{aligned}
\]
where $H(z(x))=z(x)D_m'(z(x))$, $z'(x)=S'(x)/D_m'(z)$ and
$H'(z(x))=D_m'(z(x))+z(x)D_m''(z(x))$. Since $z(x)\in(0,1)$ and
$D_m'(z(x))>0$, the condition $T''(x)\le0$ 
is equivalent to
\begin{equation}\label{eq:eq1}
     z(x)D_m'(z(x))S''(x)
        \ge
        \frac{D_m'(z(x))+zD_m''(z(x))}{D_m'(z(x))}
        (S'(x))^2.
\end{equation}
Note that
\[
        \frac{D_m'(z(x))+z(x)D_m''(z(x))}{D_m'(z(x))}
        =
        \frac{m-nz(x)}{1-z(x)},
\]
Eq.~\eqref{eq:eq1} becomes
\begin{equation}\label{eq:curvature}
        z(x)D_m'(z(x))S''(x)
        \ge \frac{m-nz(x)}{1-z(x)}(S'(x))^2.
\end{equation}
It remains to compute the first two derivatives of \(S(x)\). To avoid
notation confusion, define, for \({\bm r}=(r_1,\ldots,r_n)\in[0,1]^n\),
\[
        \Phi({\bm r})
        =
        \Pp_{\bm r}\left(\sum_{i=1}^n B_i\ge m\right),
\]
where the \(B_i\)'s are independent Bernoulli random variables with
\(
        \Pp_{\bm r}(B_i=1)=r_i.
\)
For each \(i\), conditioning on \(B_i\) gives
\[
\begin{aligned}
\Phi({\bm r})
&=
r_i\Pp_{\bm r}(N_{-i}\ge m-1)
+
(1-r_i)\Pp_{\bm r}(N_{-i}\ge m),
\end{aligned}
\]
where
\(
        N_{-i}=\sum_{\ell\ne i}B_\ell.
\)
Since \(N_{-i}\) does not involve \(r_i\),
\[
\begin{aligned}
        \frac{\partial\Phi}{\partial r_i}
        &=
        \Pp_{\bm r}(N_{-i}\ge m-1)
        -
        \Pp_{\bm r}(N_{-i}\ge m)\\
        &=
        \Pp_{\bm r}(N_{-i}=m-1).
\end{aligned}
\]
For the second derivative, for \(j\ne i\), conditioning on \(B_j\) within
\(N_{-i}\) gives
\[
\begin{aligned}
\Pp_{\bm r}(N_{-i}=m-1)
&=
r_j\Pp_{\bm r}(N_{-ij}=m-2)+
(1-r_j)\Pp_{\bm r}(N_{-ij}=m-1).
\end{aligned}
\]
Hence
\[
\begin{aligned}
\frac{\partial}{\partial r_j}
\Pp_{\bm r}(N_{-i}=m-1)
&=
\Pp_{\bm r}(N_{-ij}=m-2)-
\Pp_{\bm r}(N_{-ij}=m-1).
\end{aligned}
\]
Note that
\[
        S(x)
        =
        \Phi(q_1(x),\ldots,q_n(x)),
        \qquad
        q_i(x)=e^{-\lambda_i x}.
\]
Therefore, using
\(
        q_i'(x)=-\lambda_iq_i(x),
\)
we obtain
\[
\begin{aligned}
S'(x)
&=
\sum_{i=1}^n
\left.
\frac{\partial\Phi}{\partial r_i}
\right|_{r={\bm q}(x)}
q_i'(x)\\
&=
-\sum_{i=1}^n
\lambda_iq_i(x)
\Pp_{{\bm q}(x)}(N_{-i}=m-1).
\end{aligned}
\]
Let
\[
        \alpha_i(x)
        =
        q_i(x)\Pp_{{\bm q}(x)}(N_{-i}=m-1),
\]
and write
\(
        \bm{\alpha}(x)
        =
        (\alpha_1(x),\ldots,\alpha_n(x))^{\mathsf T}.
\)
Then
\[
        S'(x)
        =
        -\boldsymbol\lambda^{\mathsf T}\bm{\alpha}(x),
\]
where
\(
        \boldsymbol\lambda
        =
        (\lambda_1,\ldots,\lambda_n)^{\mathsf T}.
\)

Differentiating \(S'(x)\) therefore yields
\[
\begin{aligned}
S''(x)
&=
\sum_{i=1}^n
\lambda_i^2q_i(x)
\Pp_{{\bm q}(x)}(N_{-i}=m-1)\\
&\quad+
\sum_{i\ne j}
\lambda_i\lambda_jq_i(x)q_j(x)
\Big[
\Pp_{{\bm q}(x)}(N_{-ij}=m-2)
-
\Pp_{{\bm q}(x)}(N_{-ij}=m-1)
\Big].
\end{aligned}
\]
Let \(A(x)=(a_{ij}(x))\) denote the matrix in Eq.~\eqref{eq:Adef} evaluated at
\[
        q_i=q_i(x),\qquad i=1,\ldots,n.
\]
Thus
\(
        a_{ii}(x)=\alpha_i(x),
\)
and, for \(i\ne j\),
\[
\begin{aligned}
a_{ij}(x)
&=
q_i(x)q_j(x)
\Big[
\Pp_{{\bm q}(x)}(N_{-ij}=m-2) 
-
\Pp_{{\bm q}(x)}(N_{-ij}=m-1)
\Big].
\end{aligned}
\]
Consequently,
\begin{equation}\label{eq:Sderivs}
             S''(x)
        =
        \boldsymbol\lambda^{\mathsf T}A(x)\boldsymbol\lambda.
\end{equation}
We now substitute Eq.~\eqref{eq:Sderivs} into Eq.
\eqref{eq:curvature}. Since
\[
        z(x)D_m'(z(x))=m b_m(x),
\]
where
\(
        b_m(x)
        =
        \binom{n}{m}
        z(x)^m(1-z(x))^{n-m},
\)
and
$
        (S'(x))^2=
        \boldsymbol\lambda^{\mathsf T}
        \bm{\alpha}(x)\bm{\alpha}(x)^{\mathsf T}
        \boldsymbol\lambda,
$
the left-hand side of Eq.~\eqref{eq:curvature} minus the right-hand side is
\[
\begin{aligned}
&z(x)D_m'(z(x))S''(x)
-
\frac{m-nz(x)}{1-z(x)}
\{S'(x)\}^2\\
&\qquad=
m b_m(x)
\boldsymbol\lambda^{\mathsf T}
A(x)
\boldsymbol\lambda-
\frac{m-nz(x)}{1-z(x)}
\boldsymbol\lambda^{\mathsf T}
\bm{\alpha}(x)\bm{\alpha}(x)^{\mathsf T}
\boldsymbol\lambda\\
&\qquad=
m b_m(x)
\boldsymbol\lambda^{\mathsf T}
\left[
A(x)
-
\frac{m-nz(x)}
     {m b_m(x)(1-z(x))}
\bm{\alpha}(x)\bm{\alpha}(x)^{\mathsf T}
\right]
\boldsymbol\lambda.
\end{aligned}
\]
For every fixed \(x>0\),
\({\bm q}(x)
\)
and \(z(x)\) is determined by
\[
        D_m(z(x))=S(x).
\]
Therefore Proposition~\ref{prop:boundary} applies at each \(x>0\), and the
matrix in brackets is positive semidefinite. Hence, Eq.~\eqref{eq:curvature} holds. It follows that
\(
        T(x)
\)
is concave, and  the required result follows.
\end{proof}

\section{Proof of Proposition~\ref{prop:boundary}}\label{sec:matrix}

The proof of Proposition~\ref{prop:boundary} is developed through several
lemmas. We first introduce the notation
\begin{equation}\label{eq}
\mu=\sum_{i=1}^n q_i,
\qquad
d=m-\mu,
\qquad
\rho=\sum_{i=1}^n(1-q_i)\alpha_i,
\end{equation}
and define
\(
\qbar=(1-q_1,\ldots,1-q_n)^{\mathsf T}.
\)

The first lemma gives a key identity used in the positive-semidefiniteness
argument.

\begin{lemma}\label{lem:Aq}
The matrix $A$ satisfies
\begin{equation}\label{eq:Aq}
        A\qbar=d\bc.
\end{equation}
\end{lemma}

\begin{proof}
Let
$J\subset\{1,\ldots,n\}$ and  
\(
        W_J=\sum_{j\in J}B_j.
\)
 For $j\in J$, denote $W_{J,-j}=W_J-B_j$. Note that
\[
        r\Pp(W_J=r)
        =
        \sum_{j\in J}q_j\Pp(W_{J,-j}=r-1).
\]
This is because the left-hand side is the expected number of successes on the event
$\{W_J=r\}$, while the summand on the right is the probability that $B_j=1$
and $W_J=r$. Also, note that
\[
        \Ee W_J\,\Pp(W_J=r)
        =
        \sum_{j\in J}q_j\Pp(W_J=r).
\]
Subtracting the second identity from the first gives
\[
\begin{aligned}
(r-\Ee W_J)\Pp(W_J=r)
&=
\sum_{j\in J}q_j
\left[\Pp(W_{J,-j}=r-1)-\Pp(W_J=r)\right].
\end{aligned}
\]
For each $j\in J$,
\[
\begin{aligned}
\Pp(W_J=r)
&=
q_j\Pp(W_{J,-j}=r-1)
+
(1-q_j)\Pp(W_{J,-j}=r).
\end{aligned}
\]
Therefore
\[
\begin{aligned}
q_j\left[\Pp(W_{J,-j}=r-1)-\Pp(W_J=r)\right]
&=
q_j(1-q_j)
\left[\Pp(W_{J,-j}=r-1)-\Pp(W_{J,-j}=r)\right].
\end{aligned}
\]
Hence, for every integer
$r$,
\begin{equation}\label{eq:bernoulli-identity}
\sum_{j\in J}q_j(1-q_j)
\left[\Pp(W_{J,-j}=r-1)-\Pp(W_{J,-j}=r)\right]
=
(r-\Ee W_J)\Pp(W_J=r).
\end{equation}
By Eq.~\eqref{eq:Adef}, for any $i$, it hods that
\begin{align*}
    \frac{(A\qbar)_i}{q_i}
    &=(1-q_i)\Pp(N_{-i}=m-1)\\
    &\quad+\sum_{j\ne i}q_j(1-q_j)
    \left[\Pp(N_{-ij}=m-2)-\Pp(N_{-ij}=m-1)\right].
\end{align*}
Now, consider a specific case that
\[
        J=\{1,\ldots,n\}\setminus\{i\},
        \qquad
        r=m-1.
\]
Then, we have 
\(         W_J=N_{-i}
 \), $W_{J,-j}=N_{-ij}$ for $j\ne i$, and
\[
        \Ee W_J=\mu_{-i} =\sum_{j\ne i}q_j.
\]
Hence, Eq.  \eqref{eq:bernoulli-identity} becomes
\[
    \sum_{j\ne i}q_j(1-q_j)
    \left(\Pp(N_{-ij}=m-2)-\Pp(N_{-ij}=m-1)\right)
    =
    (m-1-\mu_{-i})\Pp(N_{-i}=m-1).
\]
Therefore
\begin{align*}
    (A\qbar)_i
    &=q_i(1-q_i+m-1-\mu_{-i})
      \Pp(N_{-i}=m-1)\\
    &=q_i(m-\mu)\Pp(N_{-i}=m-1)\\
    &=(m-\mu)\alpha_i
      =d\alpha_i.
\end{align*}
Since this holds for every $i$, we obtain $A\qbar=d\bc$.
\end{proof}

The next lemma gives the essential positivity on the hyperplane orthogonal
to $\bc$.

\begin{lemma}\label{lem:hyperplane}
For every ${\bm u}\in\R^n$ satisfying $\bc^{\mathsf T}{\bm u}=0$,
\begin{equation}\label{eq:hyperplane}
        {\bm u}^{\mathsf T}A{\bm u}\ge0.
\end{equation}
\end{lemma}

\begin{proof} Denote
\[
    w_i=\frac{q_i}{1-q_i},
    \qquad
    C=\prod_{i=1}^n(1-q_i),
\]
and, for the given vector ${\bm u}=(u_1,\ldots,u_n)^{\mathsf T}$, let
\[
        v_i=w_i u_i, \quad i=1,\ldots,n.
\]
Let $D_{\bm v}$ denote directional differentiation in direction
${\bm v}=(v_1,\ldots,v_n)$.
For $0\le r\le n$, let $e_r(w)$ denote the $r$th elementary symmetric
polynomial in ${\bm w}=(w_1,\ldots,w_n)$:
\[
        e_r({\bm w})
        =
        \sum_{\substack{J\subset\{1,\ldots,n\}\\ |J|=r}}
       \prod_{i\in J}w_i,
\]
with the convention $e_0({\bm w})=1$ and $e_r({\bm w})=0$ if $r<0$ or $r>n$.
For ${\bm w}_{-i}$, the same notation means that the $i$th coordinate is omitted:
\[
        e_r({\bm w}_{-i})
        =
        \sum_{\substack{J\subset\{1,\ldots,n\}\setminus\{i\}\\ |J|=r}}
        \prod_{j\in J}w_j.
\]
Since $v_i=w_i u_i$, it holds that
\[
        e_r({\bm w}+t{\bm v})
        =
        \sum_{|J|=r}
        \prod_{i\in J}w_i\prod_{i\in J}(1+tu_i).
\]
Differentiating twice at $t=0$ gives, for $1\le r\le n$,
\begin{equation}
\begin{aligned}\label{eq:D2er}
        D_v^2e_r(w)
        &=
        \sum_{|J|=r} \prod_{i\in J}w_i
        \left\{
        \left(\sum_{i\in J}u_i\right)^2
        -
        \sum_{i\in J}u_i^2
        \right\}\\
        &=F_r-G_r.
\end{aligned}
\end{equation}
where
\[
    F_r=\sum_{|J|=r} \prod_{i\in J}w_i
        \left(\sum_{i\in J}u_i\right)^2,
    \qquad
    G_r=\sum_{|J|=r} \prod_{i\in J}w_i \sum_{i\in J}u_i^2.
\]
Next, since
\[
        C=\prod_{i=1}^n(1-q_i),
        \qquad
        w_i=\frac{q_i}{1-q_i},
\]
we have
\[
\begin{aligned}
        \Pp(N_{-i}=m-1)
        &=
        \sum_{\substack{J\subset\{1,\ldots,n\}\setminus\{i\}\\ |J|=m-1}}
        \prod_{j\in J}q_j
        \prod_{\ell\notin J,\,\ell\ne i}(1-q_\ell)  \\
        &=
        \frac{C}{1-q_i}\,e_{m-1}(w_{-i}).
\end{aligned}
\]
Thus
\[
        \alpha_i
        =
        q_i\Pp(N_{-i}=m-1)
        =
        C w_i e_{m-1}({\bm w}_{-i}),
\]
or equivalently
\[
        \frac{\alpha_i}{C}=w_i e_{m-1}({\bm w}_{-i}).
\]
On the other hand,
\[
        D_{\bm v} e_m({\bm w})
        =
        \sum_{i=1}^n v_i e_{m-1}({\bm w}_{-i})
        =
        \sum_{i=1}^n w_i u_i e_{m-1}({\bm w}_{-i}).
\]
Consequently,
\begin{equation}\label{eq:alpha-app-D}
    \frac{\bc^{\mathsf T}{\bm u}}{C}
    =
    \sum_{i=1}^n u_i\frac{\alpha_i}{C}
    =
    D_v e_m({\bm w}).
\end{equation}
Thus $\bc^{\mathsf T}{\bm u}=0$ implies
\begin{equation}\label{eq:D-em-zero}
        D_{\bm v} e_m({\bm w})=0.
\end{equation}
The entries of $A$ can also be written as
\[
    \frac{a_{ii}}{C}=w_i e_{m-1}({\bm w}_{-i})
\]
and, for $i\ne j$,
\[
    \frac{a_{ij}}{C}=w_iw_j
    \left[e_{m-2}({\bm w}_{-i,-j})-e_{m-1}({\bm w}_{-i,-j})\right].
\]
Expanding the quadratic form therefore gives
\begin{equation}\label{eq:A-quadratic}
    \frac{u^{\mathsf T}Au}{C}=F_m-D_{\bm v}^2e_{m+1}({\bm w}),
\end{equation}
where $e_{n+1}\equiv0$ when $m=n$.

Let
\[
    f=e_m({\bm w}),\qquad g=e_{m+1}({\bm w}),\qquad h=\frac{g}{f}.
\]
By the Marcus--Lopes concavity theorem for elementary symmetric functions
\cite{MarcusLopes1957,Sra2020}, the ratio
$g/f=e_{m+1}/e_m$ is concave on the positive orthant.  Since $f=e_m({\bm w})$,
Eq.~\eqref{eq:D-em-zero} gives $D_{\bm v} f=0$. Hence, the first-derivative terms in
the second directional derivative of $g/f$ vanish, and
\[
    0\ge D_{\bm v}^2\!\left(\frac{g}{f}\right)
      =\frac{D_{\bm v}^2g}{f}-\frac{gD_{\bm v}^2f}{f^2}.
\]
Thus
\begin{equation}\label{eq:D2g-bound}
        D_{\bm v}^2g\le hD_{\bm v}^2f.
\end{equation}
Combining Eqs. \eqref{eq:D2er}, \eqref{eq:A-quadratic}, and
\eqref{eq:D2g-bound}, we obtain
\begin{equation}\label{eq:A-lower}
    \frac{{\bm u}^{\mathsf T}A{\bm u}}{C}
    \ge F_m-hD_{\bm v}^2f .
\end{equation}
If $0\le h\le1$, then, using $D_{\bm v}^2f=F_m-G_m$,
\[
        F_m-hD_{\bm v}^2f
        =
        (1-h)F_m+hG_m
        \ge0.
\]
If $h>1$, the concavity of $f^{1/m}$, together with
Eq.~\eqref{eq:D-em-zero}, gives $D_{\bm v}^2f\le0$. Hence
\[
        F_m-hD_{\bm v}^2f\ge0.
\]
Therefore ${\bm u}^{\mathsf T}A{\bm u}\ge0$.
\end{proof}

Combining Lemmas~\ref{lem:Aq} and~\ref{lem:hyperplane} yields the following
positive-semidefinite inequality.
\begin{lemma}
\label{prop:first-psd}
With $d$ and $\rho$ as in Eq.~\eqref{eq},
\begin{equation}\label{eq:first-psd}
    A-\frac{d}{\rho}\bc\bc^{\mathsf T}\succeq0.
\end{equation}
\end{lemma}

\begin{proof}
Let
\[
        M=A-\frac{d}{\rho}\bc\bc^{\mathsf T}.
\]
By Lemma~\ref{lem:Aq},
\[
        M\qbar
        =
        A\qbar-\frac{d}{\rho}\bc\bc^{\mathsf T}\qbar
        =
        d\bc-\frac{d}{\rho}\bc\rho
        =
        0,
\]
because $\bc^{\mathsf T}\qbar=\rho$. Now fix arbitrary ${\bm x}\in\R^n$ and set
\[
        u={\bm x}-\frac{\bc^{\mathsf T}{\bm x}}{\rho}\qbar .
\]
Then
\[
        \bc^{\mathsf T}{\bm u}
        =
        \bc^{\mathsf T}{\bm x}
        -
        \frac{\bc^{\mathsf T}{\bm x}}{\rho}\bc^{\mathsf T}\qbar
        =
        \bc^{\mathsf T}{\bm x}
        -
        \frac{\bc^{\mathsf T}{\bm x}}{\rho}\rho
        =
        0.
\]
Moreover,
\[
        {\bm x}={\bm u}+\frac{\bc^{\mathsf T}{\bm x}}{\rho}\qbar .
\]
Since $M\qbar=0$ and $M$ is symmetric,
\[
        {\bm x}^{\mathsf T}M{\bm x}
        =
        {\bm u}^{\mathsf T}M{\bm u}.
\]
Finally, because $\bc^{\mathsf T}{\bm u}=0$,
\[
        {\bm u}^{\mathsf T}M{\bm u}
        =
        {\bm u}^{\mathsf T}A{\bm u}
        -
        \frac{d}{\rho}
        \bigl(\bc^{\mathsf T}{\bm u}\bigr)^2
        =
        {\bm u}^{\mathsf T}A{\bm u}
        \ge0
\]
by Lemma~\ref{lem:hyperplane}. Therefore $M\succeq0$, which yields
Eq.~\eqref{eq:first-psd}.
\end{proof}

  By Lemma~\ref{prop:first-psd}, 
Proposition~\ref{prop:boundary} will follow once we prove
\begin{equation}\label{eq:scalar-target}
\frac{d}{\rho}
\ge
\frac{m-nz}{m b_m(1-z)}.
\end{equation}
The idea is to compare the coefficient $d/\rho$ with its value at the
homogeneous Bernoulli vector having the same tail probability. We do this by
fixing a tail level and minimizing $d/\rho$ over all Bernoulli probability
vectors with that same level.

For a Bernoulli probability vector
\(
        {\bm q}=(q_1,\ldots,q_n)\in[0,1]^n,
\)
let
\(
        S({\bm q})=\Pp_{\bm q}(N\ge m)\), and
        \(U({\bm q})=\Ee_{\bm q}(N-m)^+,
\)
where $a^+=\max\{a,0\}$. Also write
\(
        d({\bm q})=m-\sum_{i=1}^n q_i
\)
and
\begin{equation}\label{eq:rho_q}
     \rho({\bm q})
        =
        \sum_{i=1}^n q_i(1-q_i)
        \Pp_{\bm q}(N_{-i}=m-1).
\end{equation}
Since
\[
        (N-m)^+
        +
        m\mathbf 1_{\{N\ge m\}}
        =
        N\mathbf 1_{\{N\ge m\}},
\]
we have
\[
\begin{aligned}
        U({\bm q})+mS({\bm q})
        =
        \sum_{i=1}^n
        q_i\Pp_{\bm q}(N_{-i}\ge m-1).
\end{aligned}
\]
Moreover,
\[
\begin{aligned}
&\Pp_{\bm q}(N_{-i}\ge m-1)-S({\bm q})=
(1-q_i)\Pp_{\bm q}(N_{-i}=m-1).
\end{aligned}
\]
Therefore
\begin{equation}\label{eq:rho-stoploss1}
\begin{aligned}
\rho({\bm q})
&=
U({\bm q})+d({\bm q})S({\bm q}).
\end{aligned}
\end{equation}

For $\tau\in(0,1)$, define the fixed-tail level set
\[
        \mathcal L_\tau
        =
        \{{\bm q}\in[0,1]^n:S({\bm q})=\tau\}.
\]
On $\mathcal L_\tau$, define
\[
        J({\bm q})
        =
        \frac{d({\bm q})}{\rho({\bm q})}.
\]
We shall show that, for every $\tau\in(0,1)$, the minimum of $J$ over
$\mathcal L_\tau$ is attained at a homogeneous Bernoulli vector. The proof
has three parts: First, we show that a selected
minimizer may be taken to have all its interior coordinates equal. Second, a
one-sided perturbation rules out coordinates equal to zero. Third, another
one-sided perturbation rules out coordinates equal to one. The selected
minimizer must therefore be homogeneous. Applying this fixed-level result at
\[
        \tau=S({\bm q})=D_m(z)
\]
will then give the required result.

\medskip

The first lemma identifies the interior structure of a selected minimizer.

\begin{lemma}\label{lem:level-equalization}
Fix $\tau\in(0,1)$. The function $J$ has a minimizer on
$\mathcal L_\tau$ whose interior coordinates are all equal. Consequently,
there is a minimizer of the form
\[
    (\underbrace{1,\ldots,1}_{h},
     \underbrace{p,\ldots,p}_{r},
     \underbrace{0,\ldots,0}_{g}),
    \qquad 0<p<1.
\]
\end{lemma}

\begin{proof}
The set $\mathcal L_\tau$ is compact, because it is the inverse image of
$\{\tau\}$ under the continuous polynomial map
\[
        {\bm q}\longmapsto S({\bm q}),
\]
restricted to the compact cube $[0,1]^n$. Since $J({\bm q})$ has
$\rho({\bm q})$ in its denominator, we first verify that
$\rho({\bm q})>0$ throughout $\mathcal L_\tau$.

Suppose that ${\bm q}\in\mathcal L_\tau$ has $h$ coordinates equal to one
and $r$ coordinates in $(0,1)$, with the remaining coordinates equal to
zero. Since $\tau\in(0,1)$, we must have
\[
        h<m\le h+r.
\]
Indeed, $h\ge m$ would imply $S({\bm q})=1$, whereas $h+r<m$ would imply
$S({\bm q})=0$. Hence there is at least one interior coordinate $i$. For
such an $i$, the event $N_{-i}=m-1$ has positive probability, because
exactly $m-1-h$ successes can occur among the remaining $r-1$ interior
Bernoulli variables. Therefore
\[
        q_i(1-q_i)
        \Pp_{\bm q}(N_{-i}=m-1)>0,
\]
and hence
\[
        \rho({\bm q})>0.
\]
Thus $J$ is continuous on the compact set $\mathcal L_\tau$, and therefore
attains its minimum.

Let ${\bm q}\in\mathcal L_\tau$ be a minimizer of $J$, and suppose that two
of its interior coordinates satisfy
\[
        q_i=x,\qquad q_j=y,\qquad 0<x<y<1.
\]
Keep all other coordinates fixed and let
\[
        R=\sum_{\ell\ne i,j}B_\ell.
\]
Define
\[
        a=\Pp(R\ge m),\qquad
        b=\Pp(R=m-1),\qquad
        c=\Pp(R=m-2).
\]
With the remaining coordinates fixed, write $S(x,y)$ and $U(x,y)$ for the
corresponding values of $S({\bm q})$ and $U({\bm q})$. Set
\[
        u=x+y,\qquad v=xy.
\]
Conditioning on $B_i+B_j$ gives
\begin{equation}\label{eq:pair-S}
        S(x,y)=a+bu+(c-b)v.
\end{equation}
Similarly,
\begin{equation}\label{eq:pair-U}
        U(x,y)
        =
        \Ee(R-m)^+ + au+bv.
\end{equation}
The first term in Eq.~\eqref{eq:pair-U} is constant because the distribution of
$R$ is fixed. Moreover,
\[
        d({\bm q})
        =
        m-\sum_{\ell\ne i,j}q_\ell-u,
\]
so $d({\bm q})$ is affine in $u$.

On the fixed tail level $S(x,y)=\tau$, Eq.~\eqref{eq:pair-S} becomes
\[
        bu+(c-b)v=\tau-a.
\]
Thus the admissible pairs $(u,v)$ lie on a straight line. Both
$d({\bm q})$ and, by Eq.~\eqref{eq:pair-U}, $U({\bm q})$ are affine functions
of $(u,v)$. Since Eq.~\eqref{eq:rho-stoploss1} gives
\[
        \rho({\bm q})
        =
        U({\bm q})+\tau d({\bm q})
\]
on $\mathcal L_\tau$, $\rho({\bm q})$ is also affine along this line.
Consequently,
\[
        J({\bm q})
        =
        \frac{d({\bm q})}{\rho({\bm q})}
\]
is fractional linear along each connected pairwise level arc. A nonconstant
fractional-linear function is monotone on an interval and therefore cannot
attain a minimum at an interior point. Hence, if $J$ attains its minimum at
an interior point of such an arc, it must be constant along that arc.

We first rule out the case $b=c=0$. If  $b=c=0$, then
\[
        S(x,y)=a=\tau
\]
is independent of $x$ and $y$. Eq.~\eqref{eq:pair-U} becomes
\[
        U(x,y)=\Ee(R-m)^+ +\tau u.
\]
At the same time,
\[
        d({\bm q})
        =
        m-\sum_{\ell\ne i,j}q_\ell-u.
\]
Hence,  $\rho({\bm q})$ is constant. Since $x$ and $y$ are interior,
$u=x+y$ can be increased slightly while remaining feasible. This decreases
$d({\bm q})$, and therefore decreases $J({\bm q})$, contradicting
minimality. Thus $b$ and $c$ cannot both be zero.

Define
\[
        \psi(t)=a+2bt+(c-b)t^2.
\]
Since $b$ and $c$ are not both zero,
\[
        \psi'(t)
        =
        2\{b(1-t)+ct\}>0,
        \qquad 0<t<1.
\]
Furthermore,
\[
\begin{aligned}
        S(x,y)-\psi(x)
        &=
        (y-x)\{b(1-x)+cx\}>0,\\
        \psi(y)-S(x,y)
        &=
        (y-x)\{b(1-y)+cy\}>0.
\end{aligned}
\]
Hence
\[
        \psi(x)<S(x,y)<\psi(y),
\]
and there is a unique $t\in(x,y)$ such that
\begin{equation}\label{eq:equal-tail-pair}
        S(t,t)=S(x,y).
\end{equation}

Moreover,
\[
\begin{aligned}
\frac{\partial S}{\partial x}
&=
b(1-y)+cy>0,\\
\frac{\partial S}{\partial y}
&=
b(1-x)+cx>0,
\end{aligned}
\]
so $(x,y)$ is an interior point of its pairwise level arc. Since $J$
is fractional linear along this arc and has a minimum at the interior point
$(x,y)$, it must be constant along the arc. Therefore replacing $(x,y)$
by $(t,t)$ preserves both the tail level and the minimum value of $J$.

We now use this equalization property to select one minimizer whose interior
coordinates are all equal. Let
\[
        \mathcal M_\tau
        =
        \arg\min_{{\bm q}\in\mathcal L_\tau}J({\bm q})
\]
be the set of minimizers. Since $J$ is continuous and $\mathcal L_\tau$ is
compact, $\mathcal M_\tau$ is nonempty and compact. It is enough to show
that $\mathcal M_\tau$ contains one minimizer whose interior coordinates are
all equal.

Let $r_*$ be the largest number of coordinates in $(0,1)$ among vectors in
$\mathcal M_\tau$. To select a convenient representative from
$\mathcal M_\tau$, define
\[
        \Psi({\bm q})
        =
        e_{r_*}
        \bigl(q_1(1-q_1),\ldots,q_n(1-q_n)\bigr),
        \qquad {\bm q}\in[0,1]^n.
\]
Since $\mathcal M_\tau$ is compact and $\Psi$ is continuous, there exists
${\bm q}^*\in\mathcal M_\tau$ such that
\[
        \Psi({\bm q}^*)
        =
        \max_{{\bm q}\in\mathcal M_\tau}\Psi({\bm q}).
\]
By the definition of $r_*$, there is a minimizer with exactly $r_*$
interior coordinates, for which $\Psi$ is strictly positive. Any vector
with fewer than $r_*$ interior coordinates has $\Psi=0$. Therefore the
selected vector ${\bm q}^*$ has exactly $r_*$ interior coordinates. For
this vector, $\Psi({\bm q}^*)$ is simply the product of
$q_i^*(1-q_i^*)$ over its interior coordinates.

Suppose, to the contrary, that two interior coordinates of
${\bm q}^*$ satisfy
\[
        0<x<y<1.
\]
By the pairwise equalization argument above, there exists a unique
$t\in(x,y)$ such that replacing $(x,y)$ by $(t,t)$ preserves the tail
level and the value of $J$. Denote the resulting vector by
$\widetilde{\bm q}$, and then
\(
        \widetilde{\bm q}\in\mathcal M_\tau.
\)
We show that this replacement strictly increases the product of the two corresponding Bernoulli variances:
\begin{equation}\label{eq:pair-product}
        [t(1-t)]^2>x(1-x)y(1-y).
\end{equation}
If $b>0$, then \eqref{eq:equal-tail-pair} gives
\[
        b(x+y)+(c-b)xy
        =
        2bt+(c-b)t^2.
\]
Since
\[
        x+y>2\sqrt{xy},
\]
we have
\[
        S(x,y)>\psi(\sqrt{xy}).
\]
Because $S(x,y)=\psi(t)$ and $\psi$ is strictly increasing, it follows
that
\[
        t^2>xy.
\]
A direct calculation then yields
\[
\begin{split}
    \left[t(1-t)\right]^2-x(1-x)y(1-y)
    &=
    (t^2-xy)
    \left[
        (1-t)^2+\frac{c}{b}xy
    \right]>0.
\end{split}
\]
If $b=0$, then $c>0$, and \eqref{eq:equal-tail-pair} gives
\[
        t^2=xy.
\]
Therefore
\[
        \left[t(1-t)\right]^2-x(1-x)y(1-y)
        =
        xy(x+y-2\sqrt{xy})>0.
\]
Thus Eq.~\eqref{eq:pair-product} holds in either case.

The replacement preserves the number of interior coordinates, the tail
level, and the value of $J$. Hence
$\widetilde{\bm q}\in\mathcal M_\tau$ and it also has exactly $r_*$
interior coordinates. All other coordinates are unchanged, while  Eq.
\eqref{eq:pair-product} gives
\[
        \Psi(\widetilde{\bm q})
        >
        \Psi({\bm q}^*).
\]
This contradicts the choice of ${\bm q}^*$ as a maximizer of $\Psi$ over
$\mathcal M_\tau$. Therefore no two interior coordinates of
${\bm q}^*$ can be unequal. Hence all interior coordinates of the selected
minimizer are equal.
\end{proof}
 
The next lemma rules out lower-boundary coordinates. Its proof is deferred to
Appendix~\ref{app:boundary}.

\begin{lemma}\label{lem:no-zero}
Let $\tau\in(0,1)$, and let a minimizer of $J$ over $\mathcal L_\tau$ have
the form
\[
    \left(
    \underbrace{1,\ldots,1}_{h},
    \underbrace{p,\ldots,p}_{r},
    \underbrace{0,\ldots,0}_{g}
    \right),
    \qquad 0<p<1.
\]
Then $g=0$.
\end{lemma}
 
It remains to rule out upper-boundary coordinates. Its proof, based on a
tail-preserving boundary perturbation, is deferred to
Appendix~\ref{app:boundary}.

\begin{lemma}\label{lem:no-one}
Let $\tau\in(0,1)$, and let
\[
    \mathcal M_\tau
    =
    \operatorname*{arg\,min}_{{\bm q}\in\mathcal L_\tau}
    J({\bm q}).
\]
Among all vectors in $\mathcal M_\tau$, select a minimizer
${\bm q}^*$ having the largest possible number of interior coordinates.
By Lemma~\ref{lem:level-equalization}, ${\bm q}^*$ may be chosen so that
all of its interior coordinates are equal. Thus, after a permutation of
coordinates, it has the form
\[
    {\bm q}^*
    =
    \left(
    \underbrace{1,\ldots,1}_{h},
    \underbrace{p,\ldots,p}_{r},
    \underbrace{0,\ldots,0}_{g}
    \right),
    \qquad 0<p<1.
\]
If $g=0$, then $h=0$.
\end{lemma}

We now combine the three auxiliary lemmas to complete the scalar comparison.
They show that, on every fixed tail level, the minimum of $J$ is attained at
a homogeneous Bernoulli vector.

\begin{lemma}\label{lem:scalar}
Let $z$ be determined by Eq.~\eqref{eq:calibration}, and let $b_m$ be
defined in Eq.~\eqref{eq:bm}. Then
\begin{equation}\label{eq:scalar}
        \frac{d({\bm q})}{\rho({\bm q})}
        \ge
        \frac{m-nz}{m b_m(1-z)}.
\end{equation}
\end{lemma}

\begin{proof}
Fix an arbitrary tail level $\tau\in(0,1)$, and let
\[
    \mathcal M_\tau
    =
    \operatorname*{arg\,min}_{{\bm q}\in\mathcal L_\tau}
    J({\bm q}).
\]
As in the proof of Lemma~\ref{lem:level-equalization}, choose a minimizer
${\bm q}^*\in\mathcal M_\tau$ having the largest possible number of
interior coordinates among all minimizers. Lemma~\ref{lem:level-equalization}
shows that ${\bm q}^*$ may be chosen so that all of its interior coordinates
are equal. Hence, after a permutation of coordinates,
\[
    {\bm q}^*
    =
    \left(
    \underbrace{1,\ldots,1}_{h},
    \underbrace{p_\tau,\ldots,p_\tau}_{r},
    \underbrace{0,\ldots,0}_{g}
    \right),
    \qquad 0<p_\tau<1.
\]

Lemma~\ref{lem:no-zero} implies that $g=0$. Since ${\bm q}^*$ was selected
to have the largest possible number of interior coordinates among all
minimizers, Lemma~\ref{lem:no-one} then implies that $h=0$. Therefore
\[
        {\bm q}^*
        =
        {\bm q}_\tau
        =
        (p_\tau,\ldots,p_\tau),
        \qquad 0<p_\tau<1.
\]
Hence the minimum of $J$ over $\mathcal L_\tau$ is attained at a
homogeneous Bernoulli vector.

Since ${\bm q}_\tau\in\mathcal L_\tau$,
\[
        D_m(p_\tau)=\tau.
\]
We next compute the value of $J$ at this homogeneous vector. Since
\[
        d({\bm q}_\tau)=m-np_\tau,
\]
and each $N_{-i}$ has distribution
$\operatorname{Bin}(n-1,p_\tau)$,
\[
\begin{aligned}
\rho({\bm q}_\tau)
&=
np_\tau(1-p_\tau)
\Pp\left(
\operatorname{Bin}(n-1,p_\tau)=m-1
\right)\\
&=
np_\tau(1-p_\tau)
\binom{n-1}{m-1}
p_\tau^{m-1}(1-p_\tau)^{n-m}\\
&=
m\binom{n}{m}
p_\tau^m(1-p_\tau)^{n-m+1}.
\end{aligned}
\]
Therefore
\[
        J({\bm q}_\tau)
        =
        \frac{m-np_\tau}
        {m\binom{n}{m}
        p_\tau^m(1-p_\tau)^{n-m+1}}.
\]

Now take
\[
        \tau
        =
        S({\bm q})
        =
        \Pp_{\bm q}(N\ge m).
\]
By the calibration equation,
\[
        S({\bm q})=D_m(z).
\]
Since
\[
        D_m(p_\tau)=\tau=D_m(z)
\]
and $D_m$ is strictly increasing on $(0,1)$, we have
\[
        p_\tau=z.
\]
Because ${\bm q}_\tau$ minimizes $J$ over $\mathcal L_\tau$ and
${\bm q}\in\mathcal L_\tau$,
\[
        J({\bm q})
        \ge
        J({\bm q}_\tau).
\]
Using $p_\tau=z$ and
\[
        b_m
        =
        \binom{n}{m}z^m(1-z)^{n-m},
\]
we obtain
\[
\begin{aligned}
        \frac{d({\bm q})}{\rho({\bm q})}
       =
        J({\bm q})
        &\ge
        J({\bm q}_\tau)\\
        &=
        \frac{m-nz}
        {m\binom{n}{m}z^m(1-z)^{n-m+1}}\\
        &=
        \frac{m-nz}{m b_m(1-z)}.
\end{aligned}
\]
This proves Eq.~\eqref{eq:scalar}.
\end{proof}

We are now ready to prove Proposition~\ref{prop:boundary}.

\begin{proof}[Proof of Proposition~\ref{prop:boundary}]
By Lemma~\ref{prop:first-psd},
\[
    A-\frac{d({\bm q})}{\rho({\bm q})}
    \bc\bc^{\mathsf T}
    \succeq0.
\]
By Lemma~\ref{lem:scalar},
\[
    \frac{d({\bm q})}{\rho({\bm q})}
    -
    \frac{m-nz}{m b_m(1-z)}
    \ge0.
\]
Since
\(
        \bc\bc^{\mathsf T}\succeq0,
\)
it follows that
\[
\begin{aligned}
&A-\frac{m-nz}{m b_m(1-z)}
\bc\bc^{\mathsf T}=
\left(
A-\frac{d({\bm q})}{\rho({\bm q})}
\bc\bc^{\mathsf T}
\right)
+
\left(
\frac{d({\bm q})}{\rho({\bm q})}
-
\frac{m-nz}{m b_m(1-z)}
\right)
\bc\bc^{\mathsf T}
\succeq0.
\end{aligned}
\]
Hence, Proposition~\ref{prop:boundary} follows.
\end{proof}

\section{Proportional-hazards family}\label{sec:extensions}

We now extend the exponential comparison to the proportional-hazards family
defined in Eq.~\eqref{eq:family-intro}. Define
\begin{equation}\label{eq:H}
H_{\alpha,\xi}(x)
=
\begin{cases}
\displaystyle
\frac{1}{\xi}
\log\left[1+\xi(x/\sigma)^\alpha\right],
& \xi>0,\\[6pt]
\displaystyle
(x/\sigma)^\alpha,
& \xi=0.
\end{cases}
\end{equation}
Then Eq.~\eqref{eq:family-intro} can be written as
\[
        \overline F_{\lambda,\sigma,\alpha,\xi}(x)
        =
        \exp\left(-\lambda H_{\alpha,\xi}(x)\right).
\]
Consequently, if
\(
        X\sim F_{\lambda,\sigma,\alpha,\xi},
\)
then
\begin{equation}\label{eq:hazardtransform}
        H_{\alpha,\xi}(X)
        \sim
        \operatorname{Exp}(\lambda).
\end{equation}
Thus the increasing transformation \(H_{\alpha,\xi}\) maps the family in
Eq.~\eqref{eq:family-intro} to the exponential family. This allows
Theorem~\ref{thm:exp} to serve as the starting point for the extension
below. Note that, for \(\xi>0\), the survival function is regularly varying with index
\(
        -\alpha\lambda/\xi.
\)

For \(1\le k\le n\), define
\begin{equation}\label{eq:lambdak}
\lambda_{[k]}
=
\left[
\frac{e_k(\lambda_1,\ldots,\lambda_n)}
     {\binom{n}{k}}
\right]^{1/k},
\end{equation}
where \(e_k\) denotes the \(k\)th elementary symmetric polynomial.

\begin{theorem}\label{thm:unified}
Let \(X_1,\ldots,X_n\) be independent random variables from the family in
Eq.~\eqref{eq:family-intro}, with common parameters
\((\sigma,\alpha,\xi)\), heterogeneous rates
\(\lambda_1,\ldots,\lambda_n\), and \(0<\alpha\le1\). Let
\(Y_1,\ldots,Y_n\) be i.i.d. from the same family with common rate
\(\lambda_0>0\). If \(\xi=0\), then for every \(\lambda_0>0\),
\[
        Y_{k:n}\le_c X_{k:n}.
\]
If \(\xi>0\) and
\begin{equation}\label{eq:ratecondition}
        \lambda_0\ge\lambda_{[k]},
\end{equation}
then
\[
        Y_{k:n}\le_c X_{k:n}.
\]
\end{theorem}

\begin{proof}
Denote
\[
        E_i=H_{\alpha,\xi}(X_i),
        \qquad
        E_i^0=H_{\alpha,\xi}(Y_i).
\]
By Eq.~\eqref{eq:hazardtransform}, the \(E_i\)'s are independent
exponential random variables with rates
\(\lambda_1,\ldots,\lambda_n\), while the \(E_i^0\)'s are i.i.d.
exponential random variables with rate \(\lambda_0\). Define
\[
        T(s)
        =
        F_{E_{k:n}^0}^{-1}
        \left(
        F_{E_{k:n}}(s)
        \right).
\]
By Theorem~\ref{thm:exp}, \(T\) is increasing and concave, with
\(
        T(0)=0.
\)
We first derive the behavior of \(T\) at the origin. As \(s\downarrow0\),
\[
        F_{E_{k:n}}(s)
        \sim
        e_k(\lambda_1,\ldots,\lambda_n)s^k,
\]
where \(\sim\) means asymptotic equivalence, and  as \(t\downarrow0\),
\[
        F_{E_{k:n}^0}(t)
        \sim
        \binom{n}{k}(\lambda_0t)^k.
\]
It follows that
\begin{equation}\label{eq:Tprime0}
        T'(0+)
        =
        \frac{\lambda_{[k]}}{\lambda_0}.
\end{equation}
When \(\xi>0\), condition \eqref{eq:ratecondition} implies
\(
        T'(0+)\le1.
\)
Since \(T\) is increasing and concave,
\(
        0\le T'(s)\le1\),
for \(s>0.
\)

We next establish a simple power-preservation argument that will be used in
both $\xi=0$ and $\xi>0$. Let \(R\) be increasing and concave on \([0,\infty)\), with
\(R(0)=0\), and set
\[
        p=\frac{1}{\alpha}\ge1.
\]
Define
\[
        V(x)
        =
        \left[
        R\left(x^{1/p}\right)
        \right]^p.
\]
For \(x=y^p\), write
\[
        \eta(y)
        =
        \frac{yR'(y)}{R(y)}.
\]
Since \(R\) is increasing and concave with \(R(0)=0\),
\(
        0\le\eta(y)\le1.
\)
A direct calculation gives
\[
\frac{x^2V''(x)}{V(x)}
=
\frac{1}{p}
\frac{y^2R''(y)}{R(y)}
+
\left(
1-\frac{1}{p}
\right)
\left(
\eta(y)^2-\eta(y)
\right)
\le0.
\]
Hence \(V\) is concave.

Suppose first that \(\xi=0\). By Eq.~\eqref{eq:H},
\[
        H_{\alpha,0}(x)
        =
        (x/\sigma)^\alpha,
        \qquad
        H_{\alpha,0}^{-1}(s)
        =
        \sigma s^{1/\alpha}.
\]
Taking \(R=T\) in the preceding argument gives
\[
\begin{aligned}
F_{Y_{k:n}}^{-1}
\left(
F_{X_{k:n}}(x)
\right)=
H_{\alpha,0}^{-1}
\left(
T\left(
H_{\alpha,0}(x)
\right)
\right)=
\sigma
\left[
T\left(
(x/\sigma)^\alpha
\right)
\right]^{1/\alpha},
\end{aligned}
\]
which is concave. Therefore
\[
        Y_{k:n}\le_c X_{k:n}.
\]
Now suppose that \(\xi>0\). Define
\[
        g_\xi(s)
        =
        \frac{e^{\xi s}-1}{\xi}
\]
and
\[
        R(x)
        =
        g_\xi
        \left(
        T\left(
        g_\xi^{-1}(x)
        \right)
        \right).
\]
We show that \(R\) is concave. If
\(
        x=g_\xi(s),
\)
then direct differentiation gives
\[
R''(x)
=
e^{\xi(T(s)-2s)}
\left[
T''(s)
+
\xi T'(s)
\left(
T'(s)-1
\right)
\right].
\]
Because \(T\) is concave and
\(
        0\le T'(s)\le1,
\)
we have
\(
        R''(x)\le0.
\)
Thus \(R\) is increasing and concave, with \(R(0)=0\).

Moreover,
\[
        H_{\alpha,\xi}^{-1}(s)
        =
        \sigma
        \left[
        g_\xi(s)
        \right]^{1/\alpha}.
\]
Since
\[
        g_\xi
        \left(
        H_{\alpha,\xi}(x)
        \right)
        =
        (x/\sigma)^\alpha,
\]
we obtain
\[
\begin{aligned}
F_{Y_{k:n}}^{-1}
\left(
F_{X_{k:n}}(x)
\right)
=
H_{\alpha,\xi}^{-1}
\left(
T\left(
H_{\alpha,\xi}(x)
\right)
\right)=
\sigma
\left[
R\left(
(x/\sigma)^\alpha
\right)
\right]^{1/\alpha}.
\end{aligned}
\]
By the power-preservation argument above, this function is concave.
Therefore
\[
        Y_{k:n}\le_c X_{k:n}.
\]
\end{proof}

By Maclaurin's inequalities
\cite{MarshallOlkinArnold2011},
\[
        \lambda_{[k]}
        \le
        \frac{1}{n}\sum_{i=1}^n\lambda_i.
\]
Hence the arithmetic mean rate provides a single i.i.d. benchmark that works
  for all thresholds.

\begin{corollary}\label{cor:meanrate}
Let
\(
        \bar\lambda
        =
        \frac{1}{n}\sum_{i=1}^n\lambda_i.
\)
Under the assumptions of Theorem~\ref{thm:unified}, let
\(
        \lambda_0=\bar\lambda.
\)
Then
\[
        Y_{k:n}\le_c X_{k:n},
        \qquad k=1,\ldots,n.
\]
\end{corollary}

For the Weibull case \(\xi=0\), Theorem~\ref{thm:unified} shows
the comparison holds for every common positive rate \(\lambda_0\).
One may wonder if the restriction \(0<\alpha\le1\) can be dropped.
The next result  shows that  when \(\alpha>1\),  the required convex
transform ordering fails for the Weibull family.

\begin{theorem}\label{thm:alpha}
Let \(n\ge2\), \(2\le k\le n\), and \(\alpha>1\). Suppose that
\(X_1,\ldots,X_n\) are independent and satisfy
\[
        \Pp(X_i>x)
        =
        \exp\left(-\lambda_i x^\alpha\right),
        \qquad i=1,\ldots,n,
\]
where the positive rates \(\lambda_1,\ldots,\lambda_n\) are not all equal.
Let \(Y_1,\ldots,Y_n\) be i.i.d. with
\[
        \Pp(Y_i>x)
        =
        \exp\left(-\lambda_0x^\alpha\right)
\]
for any \(\lambda_0>0\). Then
\[
        Y_{k:n}\not\le_c X_{k:n}.
\]
\end{theorem}

\begin{proof} Let
\[
        m=n-k+1,
        \qquad
        E_i=X_i^\alpha,
        \qquad
        E_i^0=Y_i^\alpha.
\]
Then \(E_i\) is exponential with rate \(\lambda_i\), while
\(E_i^0\) is exponential with rate \(\lambda_0\). Define
\[
        T(s)
        =
        F_{E_{k:n}^0}^{-1}
        \left(
        F_{E_{k:n}}(s)
        \right).
\]
Then, we have
\begin{equation}\label{eq:Psi}
        \Psi(x)
        =
        \left[
        T(x^\alpha)
        \right]^{1/\alpha}.
\end{equation}
Let
\[
        \beta
        =
        \min_{\substack{I\subset\{1,\ldots,n\}\\ |I|=m}}
        \sum_{i\in I}\lambda_i,
\]
and 
\[
        M
        =
        \#\left\{
        I\subset\{1,\ldots,n\}:
        |I|=m,\ 
        \sum_{i\in I}\lambda_i=\beta
        \right\}.
\]
An expansion of the Poisson--binomial tail gives, for some \(\delta>0\),
\[
        \Pp(E_{k:n}>s)
        =
        M e^{-\beta s}
        \left[
        1+O\left(e^{-\delta s}\right)
        \right].
\]
For the homogeneous order statistic,
\[
        \Pp(E_{k:n}^0>t)
        =
        \binom{n}{m}
        e^{-m\lambda_0t}
        \left[
        1+O\left(e^{-\lambda_0t}\right)
        \right].
\]
Consequently, for some \(\eta>0\),
\begin{equation}\label{eq:Tasymp}
\begin{aligned}
        T(s)
        &=
        as+c+O\left(e^{-\eta s}\right),\\
        T'(s)
        &=
        a+O\left(e^{-\eta s}\right),\\
        T''(s)
        &=
        O\left(e^{-\eta s}\right),
\end{aligned}
\end{equation}
where
\[
        a=\frac{\beta}{m\lambda_0},
        \qquad
        c=
        \frac{1}{m\lambda_0}
        \log\left(
        \frac{\binom{n}{m}}{M}
        \right).
\]
Since \(k\ge2\), we have \(m\le n-1\). If all \(m\)-subset sums of the
rates were equal, then all rates would be equal. Hence
\(
        M<\binom{n}{m},
\)
and therefore
\(
        c>0.
\)

Denoting
\(
        s=x^\alpha,
\)
differentiation of Eq.~\eqref{eq:Psi} gives
\begin{equation}\label{eq:Psi2}
\begin{split}
\Psi''(x)
=
x^{\alpha-2}
T(s)^{1/\alpha-2}
\left[
(\alpha-1)T'(s)
\left(
T(s)-sT'(s)
\right)
+
\alpha sT(s)T''(s)
\right].
\end{split}
\end{equation}
By Eq.~\eqref{eq:Tasymp},
\[
        T(s)-sT'(s)\longrightarrow c,
\]
while
\[
        T'(s)\longrightarrow a
\]
and
\[
        sT(s)T''(s)\longrightarrow0.
\]
Therefore the expression in brackets in Eq.~\eqref{eq:Psi2} converges to
\[
        (\alpha-1)ac>0.
\]
It follows that
\(
        \Psi''(x)>0
\)
for all sufficiently large \(x\). Hence,
\[
        Y_{k:n}\le_c X_{k:n}
\]
fails.
\end{proof}

\section{Conclusion and discussion}\label{sec:discussion}
This paper closes the gap between the convex-transform comparison known for the maximum of heterogeneous exponential samples and the star-order comparison previously available for general order statistics. The main result shows that the homogeneous exponential sample provides a common distributional-shape benchmark for every order statistic: for each threshold level, heterogeneity produces a more stretched upper-quantile shape, while the homogeneous order statistic is faster aging in the convex-transform sense. The result is therefore stronger than a comparison of means, variances, or pointwise tail probabilities; it gives a scale-free comparison of the entire quantile shape. Our proof also introduces a different mechanism from those used for the maximum or for the star order. The exponential order-statistic problem is reduced to an inequality for a Poisson--binomial tail, which is then resolved through a positive-semidefinite boundary matrix and a fixed-tail extremal argument. This structure may be useful for other comparison problems in which heterogeneous Bernoulli or threshold mechanisms appear.
The proportional-hazards extension shows that the same shape comparison persists for a broader class containing Weibull, Lomax, and Burr XII distributions under explicit conditions.  These conclusions have natural interpretations across several threshold-based applications. In a \(k\)-out-of-\(n\) reliability system, an appropriate order statistic determines the time at which the required number of components has failed or remains functioning. In a recovery architecture, \(X_{k:n}\) may represent the time until \(k\) recovery resources have been restored. In voting and distributed authorization systems, it may represent the time until the required number of responses has been received.   The distinction between heterogeneous and homogeneous component times therefore remains important even when their average system times are matched by rescaling. For insurance applications, this distinction is particularly relevant to interruption and recovery risk. Our result indicates that the upper tail of the restoration-time distribution, rather than its mean alone, affects stop-loss exposure and capital requirements.

Several directions remain open. The present results rely on independence and
on a common baseline family, with heterogeneity entering through the rate
parameters. Allowing dependence among component times would be particularly
important for reliability and recovery systems subject to common shocks or
shared operational environments. It is also natural to ask whether comparable convex-transform results continue to hold when components have different shape or tail parameters, or more generally, what conditions on a baseline cumulative hazard are sufficient to preserve the comparison. More broadly, the argument developed here exploits the threshold structure of order statistics. Determining which features of the comparison survive for more general coherent systems, or for other aggregation mechanisms such as convolutions, would further clarify how heterogeneity affects distributional shape beyond the i.i.d. benchmark.

\clearpage
\appendix
\section{Boundary perturbation arguments}\label{app:boundary}

This appendix gives the detailed proofs of Lemmas~\ref{lem:no-zero} and
\ref{lem:no-one}. We use the notation introduced in Section~\ref{sec:matrix}.

\subsection{Proof of Lemma~\ref{lem:no-zero}}

\begin{proof}
Suppose \(g\ge1\). Denote
\[
        \kappa=m-h,\qquad \ell=r-\kappa+1.
\]
The \(h\) coordinates equal to one always contribute \(h\) successes, while the zero coordinates contribute none. Hence the event \(\{N\ge m\}\) is equivalent to requiring at least \(\kappa=m-h\) successes among the \(r\) interior Bernoulli variables with common success probability \(p\). Because
$$
        \tau=\Pp(N\ge m)\in(0,1),
$$
we must have
$
        1\le \kappa\le r.
$
Indeed, \(\kappa\le0\) would make the tail probability equal to one, whereas \(\kappa>r\) would make it equal to zero. Therefore
$$
        \tau
        =
        D_{r,\kappa}(p)
        =
        \Pp\left(\operatorname{Bin}(r,p)\ge\kappa\right).
$$
Now replace one zero coordinate by $\varepsilon>0$ and vary the common
interior probability to $p=p(\varepsilon)$ in order to keep the same tail
level $\tau$.  Conditional on the new coordinate, the tail probability is
\[
        (1-\varepsilon)D_{r,\kappa}(p(\varepsilon))
        +
        \varepsilon D_{r,\kappa-1}(p(\varepsilon)).
\]
Since
\[
        D_{r,\kappa-1}(p)-D_{r,\kappa}(p)
        =
        \Pp\left(\operatorname{Bin}(r,p)=\kappa-1\right),
\]
writing
\[
        C_-(p)=\Pp\left(\operatorname{Bin}(r,p)=\kappa-1\right)
\]
we have
\begin{equation}\label{eq:zero-perturb}
        D_{r,\kappa}(p(\varepsilon))
        +\varepsilon C_-(p(\varepsilon))
        =
        \tau.
\end{equation}
Since
\[
        D_{r,\kappa}'(p(0))
        =
        \frac{\ell}{1-p(0)}C_-(p(0))>0,
\]
Eq.~\eqref{eq:zero-perturb} uniquely determines $p=p(\varepsilon)$ as a differentiable function of
$\varepsilon$ for all sufficiently small $\varepsilon\ge0$.
Differentiating Eq.~\eqref{eq:zero-perturb} at $\varepsilon=0$ gives
\[
        D_{r,\kappa}'(p(0))\dot p(0)+C_-(p(0))=0.
\]
 Hence
\begin{equation}\label{eq:pprime-zero}
        \dot p(0)
        =
        -\frac{C_-(p(0))}{D_{r,\kappa}'(p(0))}
        =
        -\frac{1-p(0)}{\ell}.
\end{equation}
Along this path,
\[
        d({\bm q}(\varepsilon))=\kappa-rp(\varepsilon)-\varepsilon.
\]
Hence
\begin{equation}\label{eq:dprime-zero}
        \dot d(0)
        =
        -r\dot p(0)-1
        =
        \frac{d(0)-1}{\ell}.
\end{equation}
We next compute $\rho({\bm q})$. Along the perturbation path,
the Bernoulli probability vector is
\[
\left(
\underbrace{1,\ldots,1}_{h},
\underbrace{p(\varepsilon),\ldots,p(\varepsilon)}_{r},
\varepsilon,
\underbrace{0,\ldots,0}_{g-1}
\right).
\]
By Eq.~\eqref{eq:rho_q}, the coordinates equal to zero or one make no
contribution to $\rho({\bm q})$. Hence only the $r$ interior
coordinates and the perturbed coordinate $\varepsilon$ contribute.

Consider first one of the $r$ interior coordinates. After removing this
coordinate, the $h$ coordinates equal to one already contribute $h$
successes. Thus the remaining random coordinates must contribute exactly
\[
        m-1-h=\kappa-1
\]
successes. Conditioning on the Bernoulli variable with success probability
$\varepsilon$ gives
\[
\begin{aligned}
\Pp\left(N_{-i}=m-1\right)
&=
(1-\varepsilon)
\Pp\left(\operatorname{Bin}(r-1,p(\varepsilon))=\kappa-1\right)\\
&\quad+
\varepsilon
\Pp\left(\operatorname{Bin}(r-1,p(\varepsilon))=\kappa-2\right).
\end{aligned}
\]
Therefore, the total contribution of the $r$ interior coordinates is
\[
\begin{aligned}
&r\,p(\varepsilon)\left(1-p(\varepsilon)\right)
\Big[
(1-\varepsilon)
\Pp\left(\operatorname{Bin}(r-1,p(\varepsilon))=\kappa-1\right)\\
&\hspace{32mm}
+
\varepsilon
\Pp\left(\operatorname{Bin}(r-1,p(\varepsilon))=\kappa-2\right)
\Big].
\end{aligned}
\]
Note that
\[
        D_{r,\kappa}'(p)
        =
        r\,\Pp\left(\operatorname{Bin}(r-1,p)=\kappa-1\right)
\]
and
\[
\begin{aligned}
        C_-'(p)
        &=
        r\Big[
        \Pp\left(\operatorname{Bin}(r-1,p)=\kappa-2\right)
        -
        \Pp\left(\operatorname{Bin}(r-1,p)=\kappa-1\right)
        \Big],
\end{aligned}
\]
this contribution can be rewritten as
\[
        \omega(p(\varepsilon))
        \left[
        D_{r,\kappa}'(p(\varepsilon))
        +
        \varepsilon C_-'(p(\varepsilon))
        \right],
\]
where
\(
        \omega(p)=p(1-p).
\)
The perturbed coordinate itself contributes
\[
\begin{aligned}
&\varepsilon(1-\varepsilon)
\Pp\left(\operatorname{Bin}(r,p(\varepsilon))=\kappa-1\right)=
\varepsilon(1-\varepsilon)C_-(p(\varepsilon)).
\end{aligned}
\]
Combining the two contributions gives
\begin{equation}\label{eq:rho-zero}
\begin{split}
\rho({\bm q}(\varepsilon))
&=
\omega(p(\varepsilon))
\left[
D_{r,\kappa}'(p(\varepsilon))
+
\varepsilon C_-'(p(\varepsilon))
\right]+
\varepsilon(1-\varepsilon)C_-(p(\varepsilon)).
\end{split}
\end{equation}
Differentiating Eq.~\eqref{eq:rho-zero} with respect to $\varepsilon$ and
evaluating at $\varepsilon=0$ gives
\begin{equation}\label{eq:rhodot-zero-expanded}
\begin{split}
\left.
\frac{d}{d\varepsilon}\rho({\bm q}(\varepsilon))
\right|_{\varepsilon=0}
&=
\dot p(0)
\left[
\omega'(p(0))D_{r,\kappa}'(p(0))
+
\omega(p(0))D_{r,\kappa}''(p(0))
\right]\\
&\quad+
\omega(p(0))C_-'(p(0))
+
C_-(p(0)).
\end{split}
\end{equation}
For simplicity, write
\[
        \dot\rho(0)
        =
        \left.
        \frac{d}{d\varepsilon}\rho({\bm q}(\varepsilon))
        \right|_{\varepsilon=0}.
\]
Using
\[
\begin{split}
\omega'(p(0))D_{r,\kappa}'(p(0))
+\omega(p(0))D_{r,\kappa}''(p(0))
&=
D_{r,\kappa}'(p(0))
\left(
\kappa-(r+1)p(0)
\right),\\
\omega(p(0))C_-'(p(0))+C_-(p(0))
&=
d({\bm q}(0))C_-(p(0)),\\
\dot p(0)D_{r,\kappa}'(p(0))
&=
-C_-(p(0)),
\end{split}
\]
we obtain
\[
\begin{aligned}
\dot\rho(0)
&=
-C_-(p(0))
\left(
\kappa-(r+1)p(0)
\right)
+
d({\bm q}(0))C_-(p(0))\\
&=
p(0)C_-(p(0)).
\end{aligned}
\]
Moreover,
\[
\begin{aligned}
\rho({\bm q}(0))
&=
\omega(p(0))D_{r,\kappa}'(p(0))=
p(0)\ell C_-(p(0)).
\end{aligned}
\]
Therefore
\begin{equation}\label{eq:rhoprime-zero}
        \dot\rho(0)
        =
        \frac{\rho({\bm q}(0))}{\ell}.
\end{equation}
Similarly, along the perturbation path write
\[
        J({\bm q}(\varepsilon))
        =
        \frac{d({\bm q}(\varepsilon))}
             {\rho({\bm q}(\varepsilon))}.
\]
Hence
\[
\begin{aligned}
\left.
\frac{d}{d\varepsilon}
J({\bm q}(\varepsilon))
\right|_{\varepsilon=0}
&=
\frac{
\dot d(0)\rho({\bm q}(0))
-
d({\bm q}(0))\dot\rho(0)
}{
\rho({\bm q}(0))^2
}\\
&=
-\frac{1}
{\ell\,\rho({\bm q}(0))}
<0.
\end{aligned}
\]
Since ${\bm q}(\varepsilon)$ was constructed so that
\(
        {\bm q}(\varepsilon)\in\mathcal L_\tau
\)
for all sufficiently small $\varepsilon\ge0$, the negative derivative implies
\[
        J({\bm q}(\varepsilon))
        <
        J({\bm q}(0))
\]
for all sufficiently small $\varepsilon>0$. This contradicts the assumed
minimality of ${\bm q}(0)$ on $\mathcal L_\tau$. Therefore $g=0$.
\end{proof}

\subsection{Proof of Lemma~\ref{lem:no-one}}

 \begin{proof} Let $\mathbf q^*$ be the minimizer selected in  Lemma~\ref{lem:no-one}; thus $\mathbf q^*$ has the largest possible number of interior coordinates among all minimizers of $J$ over $\mathcal{L}_{\tau}$, and, after a permutation of coordinates, \[ \mathbf q^* = \left( \underbrace{1,\ldots,1}_{h}, \underbrace{p,\ldots,p}_{r}, \underbrace{0,\ldots,0}_{g} \right), \qquad 0<p<1. \] Suppose, to the contrary, that $h\ge1$. As in the proof of
Lemma~\ref{lem:no-zero}, we perturb one boundary coordinate while adjusting
the common interior probability so that the tail level remains equal to
$\tau$. Specifically, replace one coordinate equal to one by
$1-\varepsilon$, and let the common interior probability vary from $p(0)$
to $p(\varepsilon)$. Denote
\[
        \kappa=m-h,
        \qquad
        C(p)=\Pp\left(\operatorname{Bin}(r,p)=\kappa\right).
\]
Before the perturbation,
\[
        \tau
        =
        D_{r,\kappa}(p(0)).
\]
After the perturbation, the coordinate $1-\varepsilon$ fails with
probability $\varepsilon$. Conditional on this failure, the $r$ interior
Bernoulli variables must supply one additional success. Hence the fixed-tail
condition is
\begin{equation}\label{eq:one-perturb}
        D_{r,\kappa}(p(\varepsilon))
        -
        \varepsilon C(p(\varepsilon))
        =
        \tau.
\end{equation}
Since
\[
        D_{r,\kappa}'(p(0))
        =
        \frac{\kappa}{p(0)}C(p(0))>0,
\]
Eq.~\eqref{eq:one-perturb} determines $p(\varepsilon)$ as a differentiable
function of $\varepsilon$ for all sufficiently small $\varepsilon\ge0$.
Differentiating at $\varepsilon=0$ gives
\[
        D_{r,\kappa}'(p(0))\dot p(0)-C(p(0))=0,
\]
and therefore
\begin{equation}\label{eq:pprime-one}
        \dot p(0)=\frac{p(0)}{\kappa}.
\end{equation}
Along this path,
\[
        d({\bm q}(\varepsilon))
        =
        \kappa+\varepsilon-rp(\varepsilon).
\]
Therefore
\begin{equation}\label{eq:dprime-one}
\left.
\frac{d}{d\varepsilon}
d({\bm q}(\varepsilon))
\right|_{\varepsilon=0}
=
1-r\dot p(0)
=
\frac{d({\bm q}(0))}{\kappa}.
\end{equation}
For simplicity, write
\[
        \dot d(0)
        =
        \left.
        \frac{d}{d\varepsilon}
        d({\bm q}(\varepsilon))
        \right|_{\varepsilon=0}.
\]
We next compute the first derivative of $\rho({\bm q}(\varepsilon))$.
Using Eq.~\eqref{eq:rho_q}, the coordinates equal to zero or one make no
contribution to $\rho$. Hence only the $r$ interior coordinates and the
perturbed coordinate $1-\varepsilon$ contribute. This gives
\begin{equation}\label{eq:rho-one}
\begin{split}
\rho({\bm q}(\varepsilon))
&=
p(\varepsilon)\left(1-p(\varepsilon)\right)
\left[
D_{r,\kappa}'(p(\varepsilon))
-
\varepsilon C'(p(\varepsilon))
\right]+
\varepsilon(1-\varepsilon)C(p(\varepsilon)).
\end{split}
\end{equation}
Differentiating Eq.~\eqref{eq:rho-one} at $\varepsilon=0$ gives
\begin{equation}\label{eq:rhodot-one-expanded}
\begin{split}
\left.
\frac{d}{d\varepsilon}
\rho({\bm q}(\varepsilon))
\right|_{\varepsilon=0}
&=
\dot p(0)
\left[
\left(1-2p(0)\right)D_{r,\kappa}'(p(0))
+
p(0)\left(1-p(0)\right)D_{r,\kappa}''(p(0))
\right]\\
&\quad
-
p(0)\left(1-p(0)\right)C'(p(0))
+
C(p(0)).
\end{split}
\end{equation}
For simplicity, write
\[
        \dot\rho(0)
        =
        \left.
        \frac{d}{d\varepsilon}
        \rho({\bm q}(\varepsilon))
        \right|_{\varepsilon=0}.
\]
Using
\[
    \frac{D_{r,\kappa}''(p(0))}
         {D_{r,\kappa}'(p(0))}
    =
    \frac{\kappa-1}{p(0)}
    -
    \frac{r-\kappa}{1-p(0)}
\]
and
\[
    \frac{C'(p(0))}{C(p(0))}
    =
    \frac{\kappa}{p(0)}
    -
    \frac{r-\kappa}{1-p(0)},
\]
we obtain
\[
\begin{split}
&
\left(1-2p(0)\right)D_{r,\kappa}'(p(0))
+
p(0)\left(1-p(0)\right)D_{r,\kappa}''(p(0))
\\
&\qquad=
D_{r,\kappa}'(p(0))
\left(
\kappa-(r+1)p(0)
\right),
\end{split}
\]
and
\[
        p(0)\left(1-p(0)\right)C'(p(0))
        =
        d({\bm q}(0))C(p(0)).
\]
Together with Eq.~\eqref{eq:pprime-one},
\[
        \dot p(0)D_{r,\kappa}'(p(0))
        =
        C(p(0)).
\]
Substituting these identities into Eq.~\eqref{eq:rhodot-one-expanded} gives
\[
\begin{aligned}
\dot\rho(0)
&=
C(p(0))
\left(
\kappa-(r+1)p(0)
\right)
-
d({\bm q}(0))C(p(0))
+
C(p(0))
\\
&=
\left(1-p(0)\right)C(p(0)).
\end{aligned}
\]
Moreover,
\[
\begin{aligned}
\rho({\bm q}(0))
&=
p(0)\left(1-p(0)\right)
D_{r,\kappa}'(p(0))
\\
&=
\kappa\left(1-p(0)\right)C(p(0)).
\end{aligned}
\]
Therefore
\begin{equation}\label{eq:rhoprime-one}
        \dot\rho(0)
        =
        \frac{\rho({\bm q}(0))}{\kappa}.
\end{equation}
Using the similar idea, we have
\[
        J({\bm q}(\varepsilon))
        =
        \frac{d({\bm q}(\varepsilon))}
             {\rho({\bm q}(\varepsilon))}.
\]
Using Eqs.~\eqref{eq:dprime-one} and \eqref{eq:rhoprime-one}, we obtain
\[
\begin{aligned}
\left.
\frac{d}{d\varepsilon}
J({\bm q}(\varepsilon))
\right|_{\varepsilon=0}
&=
\frac{
\dot d(0)\rho({\bm q}(0))
-
d({\bm q}(0))\dot\rho(0)
}{
\rho({\bm q}(0))^2
}
\\
&=0.
\end{aligned}
\]
Thus the first derivative along the feasible perturbation path vanishes, so
we compute the second derivative. Differentiating Eq.~\eqref{eq:one-perturb} twice gives
\[
    D_{r,\kappa}''(p(0))\dot p(0)^2
    +
    D_{r,\kappa}'(p(0))\ddot p(0)
    -
    2C'(p(0))\dot p(0)
    =
    0.
\]
Then, it follows that
\begin{equation}\label{eq:psecond-one}
    \ddot p(0)
    =
    \frac{p(0)^2}{\kappa^2}
    \left(
        \frac{\kappa+1}{p(0)}
        -
        \frac{r-\kappa}{1-p(0)}
    \right).
\end{equation}
Next, we compute the second derivative of $\rho({\bm q}(\varepsilon))$. Since
\[
        p(1-p)D_{r,\kappa}'(p)
        =
        \kappa(1-p)C(p)
\]
and
\[
        p(1-p)C'(p)
        =
        (\kappa-rp)C(p),
\]
Eq.~\eqref{eq:rho-one} can be rewritten as
\begin{equation}\label{eq:rho-one-rewritten}
\begin{split}
\rho({\bm q}(\varepsilon))
&=
\kappa\left(1-p(\varepsilon)\right)C(p(\varepsilon))
-
\varepsilon
\left(
\kappa-rp(\varepsilon)
\right)
C(p(\varepsilon))
+
\varepsilon(1-\varepsilon)C(p(\varepsilon)).
\end{split}
\end{equation}
Differentiating Eq.~\eqref{eq:rho-one-rewritten} twice at
$\varepsilon=0$ gives
\begin{equation}\label{eq:rhosecond-expanded}
\begin{split}
\ddot\rho(0)
&=
\left.
\frac{d^2}{dp^2}
\left[
\kappa(1-p)C(p)
\right]
\right|_{p=p(0)}
\dot p(0)^2
\\
&\quad+
\left.
\frac{d}{dp}
\left[
\kappa(1-p)C(p)
\right]
\right|_{p=p(0)}
\ddot p(0)
\\
&\quad-
2
\left.
\frac{d}{dp}
\left[
(\kappa-rp)C(p)
\right]
\right|_{p=p(0)}
\dot p(0)
\\
&\quad+
2C'(p(0))\dot p(0)-2C(p(0)).
\end{split}
\end{equation}
Here
\[
        \ddot\rho(0)
        =
        \left.
        \frac{d^2}{d\varepsilon^2}
        \rho({\bm q}(\varepsilon))
        \right|_{\varepsilon=0}.
\]
The required derivatives satisfy
\[
\begin{split}
\frac{
\dfrac{d}{dp}\left[\kappa(1-p)C(p)\right]
}{
\kappa(1-p)C(p)
}
&=
\frac{\kappa}{p}
-
\frac{r-\kappa+1}{1-p},
\\
\frac{
\dfrac{d^2}{dp^2}\left[\kappa(1-p)C(p)\right]
}{
\kappa(1-p)C(p)
}
&=
\left(
\frac{\kappa}{p}
-
\frac{r-\kappa+1}{1-p}
\right)^2
-
\frac{\kappa}{p^2}
-
\frac{r-\kappa+1}{(1-p)^2},
\\
\frac{d}{dp}
\left[
(\kappa-rp)C(p)
\right]
&=
-rC(p)+(\kappa-rp)C'(p).
\end{split}
\]
Substituting these identities into Eq.~\eqref{eq:rhosecond-expanded},
together with Eqs.~\eqref{eq:pprime-one} and \eqref{eq:psecond-one}, gives
\begin{equation}\label{eq:rhosecond-one}
    \frac{\ddot\rho(0)}{\rho({\bm q}(0))}
    =
    -\frac{p(0)(r+1)}
    {\kappa^2\left(1-p(0)\right)}.
\end{equation}
On the other hand,
\[
        \ddot d(0)=-r\ddot p(0),
\]
where
\[
        \ddot d(0)
        =
        \left.
        \frac{d^2}{d\varepsilon^2}
        d({\bm q}(\varepsilon))
        \right|_{\varepsilon=0}.
\]
Combining this identity with Eqs.~\eqref{eq:psecond-one} and
\eqref{eq:rhosecond-one} gives
\begin{equation}\label{eq:second-combination}
    \ddot d(0)
    -
    d({\bm q}(0))
    \frac{\ddot\rho(0)}
         {\rho({\bm q}(0))}
    =
    -\frac{p(0)(r-\kappa)}
    {\kappa^2\left(1-p(0)\right)}.
\end{equation}
Since the first derivative of $J({\bm q}(\varepsilon))$ at
$\varepsilon=0$ is zero, its second derivative satisfies
\[
\begin{aligned}
\left.
\frac{d^2}{d\varepsilon^2}
J({\bm q}(\varepsilon))
\right|_{\varepsilon=0}
&=
\frac{1}{\rho({\bm q}(0))}
\left[
\ddot d(0)
-
d({\bm q}(0))
\frac{\ddot\rho(0)}
     {\rho({\bm q}(0))}
\right].
\end{aligned}
\]
Therefore, if $\kappa<r$,
\[
\left.
\frac{d^2}{d\varepsilon^2}
J({\bm q}(\varepsilon))
\right|_{\varepsilon=0}
=
-\frac{p(0)(r-\kappa)}
{\kappa^2
\left(1-p(0)\right)
\rho({\bm q}(0))}
<0.
\]
The path ${\bm q}(\varepsilon)$ was constructed so that
\[
        {\bm q}(\varepsilon)\in\mathcal L_\tau
\]
for all sufficiently small $\varepsilon\ge0$. Since the first derivative
vanishes and the second derivative is negative, Taylor's formula gives
\[
\begin{aligned}
J({\bm q}(\varepsilon))
&=
J({\bm q}(0))
+
\frac{\varepsilon^2}{2}
\left.
\frac{d^2}{d\varepsilon^2}
J({\bm q}(\varepsilon))
\right|_{\varepsilon=0}
+
o(\varepsilon^2)
\\
&<
J({\bm q}(0))
\end{aligned}
\]
for all sufficiently small $\varepsilon>0$. This contradicts the minimality
of ${\bm q}(0)$ on $\mathcal L_\tau$.

It remains to consider the case $\kappa=r$. In this case
\[
        D_{r,r}(p)=p^r,
        \qquad
        C(p)=p^r,
\]
so Eq.~\eqref{eq:one-perturb} becomes
\(
        \tau=(1-\varepsilon)p(\varepsilon)^r.
\)
Furthermore,
\[
        d({\bm q}(\varepsilon))
        =
        \varepsilon
        +
        r\left(1-p(\varepsilon)\right).
\]
Using Eq.~\eqref{eq:rho-one-rewritten} with $\kappa=r$ gives
\[
\begin{aligned}
\rho({\bm q}(\varepsilon))
&=
(1-\varepsilon)p(\varepsilon)^r
\left[
\varepsilon+r\left(1-p(\varepsilon)\right)
\right]=
\tau\,d({\bm q}(\varepsilon)).
\end{aligned}
\]
Hence
\[
J(\mathbf q(\varepsilon))
=
\frac{d(\mathbf q(\varepsilon))}
{\rho(\mathbf q(\varepsilon))}
=
\frac{1}{\tau}
\]
is constant along this path. Therefore, for every sufficiently small
$\varepsilon>0$, the perturbed vector $\mathbf q(\varepsilon)$ is
another minimizer of $J$ over $\mathcal{L}_{\tau}$. The coordinate
formerly equal to one has become $1-\varepsilon\in(0,1)$, while the
original $r$ interior coordinates remain in $(0,1)$. Thus
$\mathbf q(\varepsilon)$ has $r+1$ interior coordinates.
This contradicts the defining property of $\mathbf q^*$, which has
the largest possible number of interior coordinates among all
minimizers of $J$ over $\mathcal{L}_{\tau}$. Therefore $h=0$.

\end{proof}


\bibliographystyle{plainnat}
\bibliography{order_statistics_aap}

\end{document}